%% file: main.tex
\documentclass[10pt,a4paper]{article}
\usepackage{amssymb}
\usepackage{amsmath}

\input{macro}

\makeatletter

\def\ps@headings{
	\def\@oddhead{		}
	\def\@evenhead{}
	\def\@evenfoot{\hfil \,\,\thepage\hfill}
	\def\@oddfoot{\hfil \,\,\thepage\hfil}
}

\def\ps@titleheadings{
	\def\@oddhead{}
	\def\@evenhead{}
	\def\@evenfoot{\hfil \,\,\thepage\hfill}
	\def\@oddfoot{\hfil \,\,\thepage\hfil}
}
\makeatother

\title{
	\textbf{Misère Greedy Nim and \\ Misère Bounded Greedy Nim }}
\author{
	Nanako Omiya\thanks{Tohoku University} \and
	Ryo Yoshinaka\footnotemark[1] \and
	Ayumi Shinohara\footnotemark[1]
}
\date{\empty}

\begin{document}
\maketitle

\input{docs/abstract.tex}
\input{docs/introduction.tex}

\input{docs/preliminary.tex}

\input{docs/Bounded1.tex}
\input{docs/Bounded2.tex}
\input{docs/Greedy.tex}
\input{docs/conclusion.tex}


\bibliographystyle{plain}
\bibliography{ref}

\end{document}

%% file: macro.tex
\newif\ifdraft \drafttrue   

\ifdraft
\newcommand{\shinocom}[1]{\textcolor{magenta}{[(篠原)#1]}}

\newcommand{\usercom}[1]{\textcolor{blue}{[(user)#1]}}
\newcommand{\todo}[1]{{\color{red}{[ToDo: #1]}}}
\else
\newcommand{\shinocom}[1]{}
\newcommand{\yoshicom}[1]{}
\newcommand{\usercom}[1]{}
\newcommand{\todo}[1]{}
\fi

\usepackage{color}                      
\usepackage{graphicx}                
\usepackage{float}
\usepackage{enumitem}
\usepackage{amsthm}
\usepackage{bm}

\setlist[itemize]{topsep=4pt, itemsep=2pt, parsep=1pt,leftmargin=12pt}
\setlist[enumerate]{topsep=4pt, itemsep=2pt, parsep=1pt,leftmargin=16pt}

\newcommand{\next}{\mathit{next}}

\theoremstyle{definition}
\newtheorem{definition}{Definition}
\theoremstyle{plain}
\newtheorem{theorem}{Theorem}

\newtheorem{lemma}{Lemma}

\newtheorem{fact}{Fact}

\newcommand{\insertsep}[2]{%
  \def\result{}%
  \def\sep{#2}%
  \insertsephelper#1\relax
  \result
}
\def\insertsephelper#1{%
  \ifx#1\relax
  \let\next\relax
  \else
    \ifx\result\empty
      \edef\result{#1}%
    \else
      \edef\result{\result\sep#1}%
    \fi
    \let\next\insertsephelper
  \fi
  \next
}
\newcommand{\cnum}[1]{{}\raise0.2ex\hbox{\textcircled{\scriptsize$\insertsep{#1}{\!\;\!}$}}}

%% file: docs/Abstract.tex
\begin{abstract}
    In this paper, we analyze the misère versions of two impartial combinatorial games: $k$-Bounded Greedy Nim and Greedy Nim. 
    We present a complete solution to both games by showing necessary and sufficient conditions for a position to be $P$-positions.
\end{abstract}

%% file: docs/introduction.tex
\section{Introduction}

A game in which two players take turns making moves without any chance elements or hidden information is called a \emph{combinatorial game}. 
Furthermore, a combinatorial game in which both players have the same possible moves on the game board is called an \emph{impartial game}.

One of the most famous impartial games is Nim~\cite{bouton1901nim}.
Nim is a stone removal game where players take turns removing stones from a heap. 
In the initial state, there are one or more heaps of stones on a table. 
The two players take turns removing at least one stone from one of the heaps, and the player who cannot move loses the game.
Impartial games can theoretically be analyzed, and there is always a winning strategy. 

Since Bouton's complete analysis of Nim in 1902~\cite{bouton1901nim}, 
several variants have been proposed and studied in the literature. 
For example, $k$-Bounded Nim, proposed by Schwartz~\cite{Schwartz01111971}, restricts the number of stones removed on each move to be less than or equal to a given constant $k$.
In Greedy Nim, introduced by Albert and Nowakowski~\cite{albert2004nim}, players must always remove stones from the heap with the most stones.
Wythoff Nim (Wythoff's Game)~\cite{wythoffgame} requires players to either take at least one stone from one heap or the same number of stones from both heaps.
The Maya Game (Welter's Game)~\cite{welter1952advancing} is a variant of Nim that prohibits making heaps with the same number of stones. 
In addition, some variants allow splitting or merging heaps. 
For instance, Kayles~\cite{dudeney2002canterbury} allows players to remove one or two stones and then split the heap into two. 
In Amalgamation Nim~\cite{locke2021amalgamation}, players merge two heaps instead of removing stones.
In 2018, Xu and Zhu~\cite{xu2018bounded} proposed $k$-Bounded Greedy Nim, which is a combination of $k$-Bounded Nim and Greedy Nim.

Impartial games like Nim, in which the player who cannot make a move loses, are called \emph{normal play} games.
Conversely, games in which the player unable to make a move wins are called \emph{misère play} games. 
Misere Nim was also fully analyzed by Bouton. 
However, it is generally known that misère games are more difficult to analyze than normal games~\cite{Grundy_Smith_1956}.
For example, Dawson's Chess~\cite{dawson1935caissa} has been completely analyzed in its normal play, 
but a complete misère analysis has not been found~\cite{PLAMBECK2008593}. 

A powerful approach to analyzing such games involves using \emph{misère quotients}, a specialized commutative monoid and its homomorphisms, as proposed by Plambeck and Siegel~\cite{PLAMBECK2008593}.
In impartial games, positions can sometimes be divided into multiple independent subgames. During their turn, a player selects one of these subgames to make a move, leaving the others unchanged.
For example, in standard Nim, a position $S=(3,2,1)$ can be divided into three independent positions:
$X=(3)$, $Y=(2)$, $Z=(1)$.
To use misère quotients, positions can be divided into independent subgames.
However, this is not possible in the case of Greedy Nim because its rules prevent dividing a position into independent positions.
In Greedy Nim, a player must remove stones from the heap with the largest number of stones, meaning they cannot arbitrarily select a subgame to act upon.
For instance, in $S=(3,2,1)$, it is not possible to make a move affecting only $Y=(2)$, as doing so would violate the game's rules.
Similarly, $k$-Bounded Greedy Nim cannot be analyzed using misère quotients for the same reason. 

Yamasaki~\cite{yamasaki1980misere} analyzed various impartial games in misère play, such as Sato's Maya Game (Welter's game)~\cite{welter1952advancing} 
and Wythoff Nim~\cite{wythoffgame}, by determining the so-called singular positions and standard positions, defined as follows:
\begin{definition}
  \label{def1}
  A position $S$ is a \emph{singular position} if its outcome differs between normal and misère play. 
  Conversely, a \emph{standard position} is a position where the outcome remains the same in both versions.
\end{definition}

If a game is completely analyzed in its normal play, determining all singular and standard positions is equivalent to completely analyzing it in its misère play.
For Nim, it is known that position $(x_1,x_2,\dots,x_n)$ where $x_1 \leq 1$ are singular positions, while positions where $x_1 \geq 2$ are standard positions~\cite{bouton1901nim}.
However, especially for Nim variants that allow heap splitting, such as Kayles~\cite{dudeney2002canterbury}, determining singular positions is significantly more challenging.

In this paper, we solve misère $k$-Bounded Greedy Nim and misère Greedy Nim by giving a characterization of all the $P$-positions for these games. 



%% file: docs/preliminary.tex
\section{Preliminaries}

\subsection{Position}

A collection of heaps of stones is represented as a sequence of non-negative integers, $S = (x_1, x_2, \dots, x_n)$. 
This means that there are $n$ heaps on the game board, and the $i$th heap from the left contains $x_i$ stones.
Note that we allow $x_i$ to be zero.
Throughout this paper, we assume without loss of generality that $n \ge 4$.
In addition, we always arrange the sequence in descending order such that $x_1 \geq x_2 \geq \dots \geq x_n$.
We call such a sequence a \emph{position}.

\subsection{N-Position and P-Position}
A position $S$ can transition to another position $S'$ by performing a single valid move, and such a position $S'$ is called a \emph{follower position} of $S$.
In impartial games, each position can be classified into one of two categories: N-positions and P-positions, which together form what is known as the \emph{outcome}.
An \emph{N-position} is a position such that the next player has a winning strategy.
On the other hand, a position is a \emph{P-position} if the previous player has a winning strategy.
We denote by $\mathbb{P}$ the set of all P-positions. 
In Nim, it is well known that a position is a P-position if and only if the bitwise XOR of the numbers of stones in all heaps is zero~\cite{bouton1901nim}. 
This XOR operation is computed by converting the number of stones in each heap into binary and summing them without carrying over.

In normal play, the position $S_0=(0,0,\dots,0)$ with no stones is a P-position.
In misère play, however, the position $S_0$ is an N-position. 
For any other position, if there exists at least one P-position among its follower positions, then the current position is an N-position. 
Conversely, if all follower positions are N-positions, then the current position is a P-position.
Thus, whether a position is an N-position or a P-position can be determined inductively.
However, when the numbers of heaps and stones are large, determining all reachable positions requires a huge amount of computational time.
This study aims to reveal the properties of all P-positions in misère $k$-Bounded Greedy Nim and misère Greedy Nim, 
allowing for an efficient linear classification of positions as N-positions or P-positions. 
Additionally, it seeks to identify the optimal strategy for N-positions.

%% file: docs/Bounded1.tex
\section{Misère $k$-Bounded Greedy Nim}
\subsection{The Rule of Misère $k$-Bounded Greedy Nim}
Let $k$ be a fixed positive integer. 
The rules of misère $k$-Bounded Greedy Nim are as follows. 
The starting position consists of several heaps of stones, with each heap containing an arbitrary number of stones. 
Two players take turns, and on each turn, a player must remove at most $k$ stones from the heap with the largest number of stones. 
The player who cannot make a move wins the game.






\subsection{Solution for $k$-Bounded Greedy Nim}
If $k=1$, exactly one stone is removed in each turn. Therefore, the following theorem holds.
\begin{theorem}
\label{theorem1}
A position $S$ is a P-position for normal play $1$-Bounded Greedy Nim if and only if the total number of all the stones is even. 
Similarly, a position $S$ is a P-position for misère $1$-Bounded Greedy Nim if and only if the total number of all the stones is odd.
\end{theorem}

In the following, we assume $k\geq 2$.

For normal play $k$-Bounded Greedy Nim, Xu and Zhu~\cite{xu2018bounded} provided a complete analysis, proving Theorem~\ref{theorem2}.

\begin{definition}[\cite{xu2018bounded}]
\label{def2}
For any position $S=(x_1,x_2,\dots,x_n)$, define
  \begin{equation}
    \beta(S)= 
        \begin{cases}
            0              &  \text{if $x_3=0$},\\
            \text{max}\{j:x_j=x_3\}-2 &  \text{if $x_3 \neq 0$}, \nonumber
        \end{cases}
    \end{equation}
    which is the number of  repetitions of the value of $x_3$ in the sequence $(x_3,x_4,\dots,x_n)$ unless $x_3=0$.
\end{definition}
\noindent

For any integer $a$, let $R(a)$ denote the remainder when $a$ is divided by $k+1$.
\begin{definition}[\cite{xu2018bounded}]
\label{def3}
  Let $x_1,x_2,x_3$ be positive integers, where $x_1 \geq x_2\geq x_3$.
  We say the triple $(x_1,x_2,x_3)$ is \emph{$k$-good} if one of the following holds:
\begin{enumerate}[label=(\arabic*),leftmargin=24pt]
\item $R(x_1-x_2 )=k $ and $R(x_2-x_3)=0$,
\item $R(x_1-x_2 )=0$ and $1 \leq R(x_2-x_3) \leq k-1$,
\item $R(x_1-x_2 )=1$ and $R(x_2-x_3)=k$.
  \end{enumerate}
\end{definition}   

\begin{theorem}
\label{theorem2}
 A position $S=(x_1,x_2,\dots,x_n)$ is a P-position in normal play $k$-Bounded Greedy Nim if and only if one of the following hold:
\begin{itemize}
\item $\beta(S)$ is even and $R(x_1-x_2)=0$,
\item $\beta(S)$ is odd and $(x_1,x_2,x_3)$ is $k$-good.
\end{itemize}
\end{theorem}

In this paper, we prove the necessary and sufficient conditions for a position to be a P-position in misère  $k$-Bounded Greedy Nim, 
as follows:
\begin{definition}
  \label{def4}
    Let $x_1,x_2$ be positive integers, where $x_1 \geq x_2$.
    We say $(x_1,x_2)$ is \emph{$k$-nice} if one of the following holds:
    \begin{enumerate}[label=(\arabic*),leftmargin=24pt]
\item $R(x_1)=0 $ and $R(x_2)=1$,
\item $R(x_1-x_2)=0$ and $R(x_2) \geq 2$,
\item $R(x_1)=1$ and $R(x_2)=0$.
    \end{enumerate}
  \end{definition}

We now present the main result for misère $k$-Bounded Greedy Nim.
\begin{theorem}
\label{theorem3}
  A position $S=(x_1,x_2,\dots,x_n)$ is a P-position in misère $k$-Bounded Greedy Nim position
  if and only if one of the following hold:
  \begin{itemize}
  \item $x_3\leq1$, $\beta(S)$ is even, and $(x_1,x_2)$ is $k$-nice.
  \item $x_3\leq1$, $\beta(S)$ is odd, and $R(x_1-x_2)=0$.
  \item $x_3\geq 2$, $\beta(S)$ is even, and $R(x_1-x_2)=0$.
  \item $x_3\geq 2$, $\beta(S)$ is odd, and $(x_1,x_2,x_3)$ is $k$-good.
  \end{itemize}
\end{theorem}


Throughout this paper, we fix a position $S=(x_1,x_2,\dots,x_n)$ and a follower position $S'=(x_1',x_2',\dots,x_n')$ after removing $t$ stones, where $1\leq t \leq \min\{x_1,k\}$.
It should be noted that in position $S'$, the condition $x_1'\geq x_2'\geq \dots \geq x_n'$ always holds.
For example, if the starting position is $S=(4,3,1)$ and two stones are removed, the position is $S'=(3,2,1)$.

The proof of Theorem \ref{theorem3} is by induction on the total number of stones.
First, in Section 3.2.1, the proof is carried out for the case where the total number of stones is less than or equal to 1.
 Next, in Section 3.2.2, a lemma useful in subsequent proofs is stated. 
 Then, in Section 3.2.3, the case where $x_3\leq1$ is considered, and in Section 3.2.4, the case where $x_3\geq2$ is considered, with the necessity and sufficiency of each case being proven.

\subsubsection{Case: Total Number of Stones Is Less than or Equal to 1}\label{sec:base}
For cases where the total number of stones is 1 or fewer, Theorem~\ref{theorem3} holds trivially.
If the total number of stones is 0, then the position $S$ is an N-position.
The position $S$ satisfies $x_3\leq 1$ and $\beta(S)$ is even.
Also, since $R(x_1)=R(x_2)=0$, $(x_1,x_2)$ is not $k$-nice. 

If the total number of stones is 1, then the position $S$ is a P-position.
The position $S$ satisfies $x_3\leq 1$ and $\beta(S)$ is even.
Also, since the position $S$ satisfies $R(x_1)=0$ and $R(x_2)=1$, $(x_1,x_2)$ satisfies (3) of the definition of $k$-nice. 

This completes the base cases for the induction.
In the remainder of the proof, we assume that the total number of stones is at least 2.

\subsubsection{Stable Moves}
This subsection introduces \emph{stable moves} and shows some properties of such moves.
The notion makes our arguments concise. 

\begin{definition}[stable move]
  \label{def5}
  We define removing $t$ stones as a \emph{stable move} if the condition $x_1-x_2\geq t$ holds. 
\end{definition}

\begin{fact}
  \label{claim1}
    If removing $t$ stones is a stable move, then the following holds for the follower position $S'$:
    \begin{itemize}
     \item $x_1'=x_1-t$ and $x_j'=x_j$ for $2\leq j\leq n$,
     \item $\beta(S')=\beta(S)$.
    \end{itemize}
\end{fact}


\begin{fact}
  \label{claim2}
  If removing $t$ stones is not a stable move, then $x_1'=x_2$.
\end{fact}

\begin{fact}
\label{claim4}
  If $r_1 =  R(x_1 - x_2) \neq 0$, removing $r_1$ stones is a stable move and $R(x_1' - x_2') = 0$ holds.
\end{fact}

  

\begin{lemma}
  \label{claim3}
  If removing $t$ stones is a stable move and $S$ satisfies the conditions of Theorem~\ref{theorem3} for being a P-position, 
  the follower position $S'$ does not satisfy them.
\end{lemma}
\begin{proof}
  If removing $t$ stones is a stable move, the new position is $S'=(x_1-t,x_2,\dots,x_n)$ and $\beta(S)=\beta(S')$ by Fact~\ref{claim1}. 
  Since $R(x_2')=R(x_2)$ holds but $R(x_1-t)\neq R(x_1)$, and $R(x_2'-x_3')=R(x_2-x_3)$ holds but $R(x_1'-x_2')=R(x_1-t-x_2)\neq R(x_1-x_2)$,
  it follows that $S'$ does not satisfy the conditions of Theorem~\ref{theorem3} if $S$ does.
\end{proof}

Hereafter, we write $r_1=R(x_1-x_2)$ for simplicity.

We split the proof into two cases, based on the value of $x_3$.

\subsubsection{Proof: Case $x_3\leq1$}
In the case when $x_3\leq1$, the state transition diagram is shown in Figure~\ref{fig:fig1}. 
The blue states represent P-positions, while the red states represent N-positions. 
While we show all possible transitions from each P-position, we draw only optimal transitions from each N-position.
 The proof proceeds in the order indicated by the numbered arrows in the diagram.
The conditions for case distinction are followed by their corresponding arrow numbers.
We first show the necessity.
\begin{figure}[H]
  \includegraphics[width=1.0\columnwidth, clip]{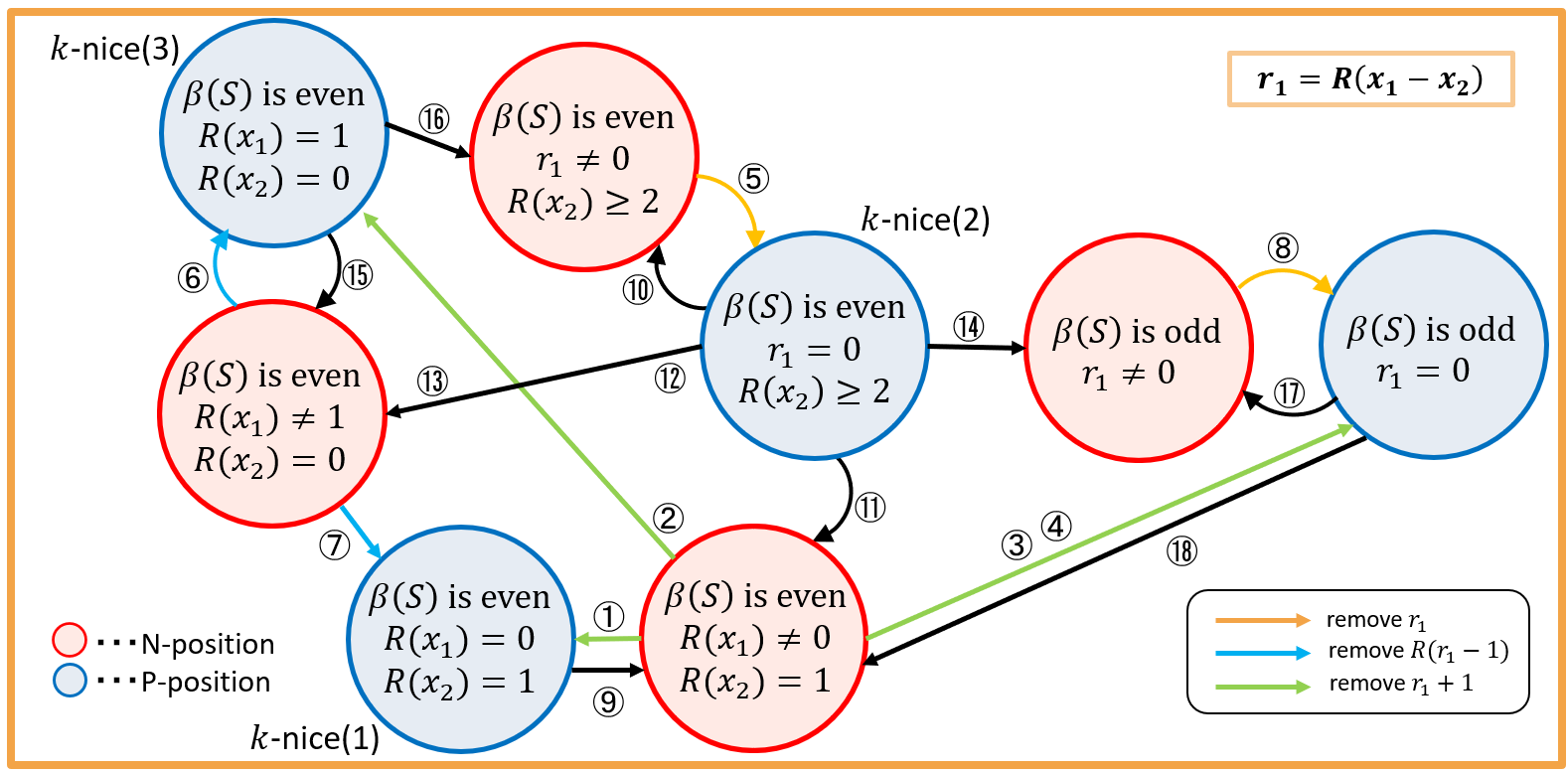}
  \caption{The state transition diagram when $x_3\leq1$}
  \label{fig:fig1}
  \end{figure}

\paragraph{Necessity.}
(The position $S$ is a P-position $\implies$ The position $S$ satisfies one of the conditions of Theorem~\ref{theorem3})

\vspace{4pt} \noindent
\textbf{Case 1: $\beta(S)$ is even.}

\noindent
We first consider the case where $\beta(S)$ is even. 
We will prove that $(x_1,x_2)$ is $k$-nice if the position $S$ is a P-position, using contraposition.
We will show that removing $t$ stones from a position $S$ that is not $k$-nice is a winning strategy where $t$ is determined below.
Here, we verify that $1\leq t\leq \min \{x_1,k\}$.
\begin{enumerate}
  \item If $R(x_1) \neq 0$ and $R(x_2) = 1$, then $t = r_1 + 1$.  

        Since $0 \leq r_1 \leq k$ and the assumption ensures $r_1 \neq k$, it follows that $1 \leq r_1 + 1 \leq k$.  
        Additionally, $x_1 \geq r_1 + x_2 \geq r_1 + 1$ holds.  
        Therefore, $1\leq t\leq \min \{x_1,k\}$ holds.
  \item If $r_1 \neq 0$ and $R(x_2) \geq 2$, then $t = r_1$.  

        Since $0 \leq r_1 \leq k$ and the assumption ensures $r_1 \neq 0$, it follows that $1 \leq r_1 \leq k$.  
        Additionally, $x_1 \geq r_1 + x_2 \geq r_1 + 2$ holds.
        Therefore, $1\leq t\leq \min \{x_1,k\}$ holds.
 \item If $R(x_1) \neq 1$ and $R(x_2) = 0$, then $t = R(r_1 - 1)$.  
        
        Since the assumption ensures $r_1 \neq 1$, it follows that $1 \leq R(r_1 - 1) \leq k$.  
        Additionally,
        if $x_2=0$, $R(r_1-1)=R(x_1-1) \leq x_1-1$ by $x_1 \neq 0$.
        On the other hand, if $x_2\neq 0$, then $x_1 \geq r_1 + x_2 \geq r_1+k+1 \geq R(r_1-1)$ holds.
        Therefore, $1\leq t\leq \min \{x_1,k\}$ holds.
\end{enumerate}
   
   Next, we prove that the position $S'$ is a P-position.
   
   \begin{enumerate}
       \item When $R(x_1) \neq 0$ and $R(x_2) = 1$:  \\
             In this case, since $R(x_1) \geq R(x_2)$, we have
             $R(x_1 - x_2) = R(x_1) - R(x_2) = R(x_1) - 1.$
             Also, $R(x_1 - t) = R(x_1 - (r_1 + 1)) = R(x_1 - (R(x_1) - 1 + 1)) = R(x_1 - R(x_1)) = 0.$
             We now consider four cases.
             \begin{itemize}
                 \item When $x_1 - t \geq x_2$:  \cnum{1}\\
                       Since the move is stable,
                       $x'_1 = x_1 - t$, $x'_j = x_j$ for all $j \geq 2$, and $\beta(S') = \beta(S)$ is even.
                       Since $R(x'_1) = R(x_1 - t) = 0$ and $R(x'_2) = R(x_2) = 1$,
                       $(x'_1, x'_2)$ satisfies  (1) of the definition of $k$-nice.  
                       By the induction hypothesis, we conclude that $S' \in \mathbb{P}$.
                 \item When $x_2 > x_1 - t \geq 1$:  \cnum{2}\\
                       Since the move is not stable, $x'_1 = x_2$.
                       Also since $x_3 \leq 1$, we have $x'_2 = x_1 - t$ and $x_j'=x_j$ for all $j \geq 3$. 
                       So $\beta(S') = \beta(S)$ remains even.  
                       Additionally, $R(x'_1) = R(x_2) = 1$ and $R(x'_2) = R(x_1 - t) = 0$,  
                       so $(x'_1, x'_2)$ satisfies (3) of the definition of $k$-nice.  
                       By the induction hypothesis, $S' \in\mathbb{ P}$.
                \item When $x_1 - t = x_3=0$: \cnum{3}\\
                      Since $x_3 = 0$, the starting position $S$ has just two non-zero heaps.
                      All the stones are removed from the leftmost heap of $S$, 
                       resulting in just one non-zero heap with $x_1' = x_2$ stones in $S'$.
                       Since $R(x'_1) = R(x_2) = 1$ and $R(x'_2) = 0$, $(x'_1, x'_2)$ satisfies (3) of the definition of $k$-nice. 
                       By the inductive hypothesis, $S' \in \mathbb{P}$.
                \item When $x_1 - t =0$ and $ x_3=1$: \cnum{4}\\
                      All the stones are removed from the leftmost heap,
                      so $x'_j = x_{j+1}$ for all $j \geq 1$.
                      Moreover, since $\beta(S)$ is even and non-zero, $\beta(S') = \beta(S) - 1$ is odd.
                      Here, given $R(x_2) = 1$ and $x_3 = 1$, $R(x_2 - x_3) = 0$. 
                      Therefore, $R(x'_1 - x'_2) = R(x_2 - x_3) = 0$.
                      By the induction hypothesis, $S' \in \mathbb{P}$.
             \end{itemize}
       \item When $R(x_1-x_2)\neq 0$ and $R(x_2)\geq 2$: \cnum{5}\\
       By Fact~\ref{claim4}, after removing $r_1=R(x_1-x_2)$ stones, the position $S'$ satisfies  $ R(x'_1 - x'_2) = 0$ and
       $\beta(S') = \beta(S)$ is even.
       Moreover, since $R(x'_1 - x'_2) = 0$ and $R(x'_2) = R(x_2) \geq 2$, $(x'_1, x'_2)$ satisfies (2) of the definition of $k$-nice.
       Thus, by the induction hypothesis, $S' \in \mathbb{P}$.

       \item When $R(x_1) \neq 1$ and $R(x_2) = 0$:\\
       We obtain that $R(x_1 - t) = R(x_1 - R(r_1 - 1)) = R(x_1 - (x_1 - x_2 - 1)) = R(x_2 + 1) = 1$.
       \begin{itemize}
        \item When $x_1>x_2$: \cnum{6}\\
        If $r_1 = 0$, then $x_1 \geq x_2 + k+1$ and thus the move is stable.
        If $r_1 \neq 0$, then $x_1 - R(r_1 - 1) = x_1 - (r_1 - 1) \ge x_2 + 1$ and thus the move is stable.
        Hence, $x_1'=x_1 - t$, $x'_j = x_j$ for all $j \geq 2$, and $\beta(S') = \beta(S)$ is even.
        Since $R(x'_1) = R(x_1 - t) = 1$ and $R(x'_2) = R(x_2) = 0$ hold,
        $(x'_1, x'_2)$ satisfies  (3) of the definition of $k$-nice.
        By the induction hypothesis, we conclude that $S' \in \mathbb{P}$.
       \item When $x_1 = x_2$: \cnum{7}\\
       Removing $t$ stones is not a stable move, so $x'_1 = x_2$.
       Moreover, since $R(x_1) = R(x_2) = 0$, it follows that $x_1$ is a multiple of $k + 1$ and $x_1 - k \geq 1$.
       Since $x_3 \leq 1$, we have $x'_2 = x_1 - t$.  
       Additionally, $ \beta(S') = \beta(S) $ is even.
       Since  $ R(x'_1) = R(x_2) = 0 $  and  $ R(x'_2) = R(x_1 - t) = 1 $,  
       $(x'_1, x'_2)$ satisfies  (1) of the definition of $k$-nice.
        By the induction hypothesis, we conclude that $S' \in \mathbb{P}$.
       \end{itemize}
   \end{enumerate}

\vspace{4pt} \noindent
\textbf{Case 2: $\beta(S)$ is odd.} \cnum{8}

\noindent
We now consider the case where $\beta(S)$ is odd. 
We prove by contraposition that if the position $S$ is a P-position, then $r_1 = R(x_1 - x_2) = 0$.
By Fact~\ref{claim4}, if $r_1 \neq 0$, then after removing $r_1$ stones, the position $S'$ satisfies $R(x'_1 - x'_2) = 0$ and $\beta(S') = \beta(S)$ is odd.
By the induction hypothesis, $S' \in \mathbb{P}$.  

Now, we proceed to the proof of the sufficiency.
\paragraph{Sufficiency.}
(The position  $S$ satisfies one of the conditions of Theorem \ref{theorem3} $\implies$ The position $S$ is a P-position)

If $S$ satisfies one of the conditions of Theorem~\ref{theorem3} and the move is stable, by Lemma~\ref{claim3} and the induction hypothesis, $S'$ is an N-position and thus $S$ is a P-position.
Therefore, it suffices to discuss the case where the move is not stable, i.e., $t > x_1-x_2$.

\vspace{4pt} \noindent
\textbf{Case 1: $\beta(S)$ is even.}  

\noindent
We first consider the case where $\beta(S)$ is even.
We will prove that if $(x_1, x_2)$ is $k$-nice, the position $S$ is a P-position. 
We will show that $S'\notin \mathbb{P}$ if $(x_1, x_2)$ satisfies (1), (2) or (3) of the definition of $k$-nice, in that order.
\begin{enumerate}
\item When $R(x_1) = 0$ and $R(x_2) = 1$: \cnum{9}\\
  In this case, $x_1-x_2 \geq k$ always holds, so the move is stable. 
\item When $R(x_1-x_2)=0$, i.e., $x_1 = x_2$, and $R(x_2) \geq 2$:  
\begin{itemize}
\item When $x_1 - t \geq 1$:\\
  Removing $t$ stones is not a stable move, so $x'_1 = x_2$.
  Since $x_3 \leq 1$, $x'_2 = x_1 - t$ holds.
  Moreover, since $x_j' = x_j$ for all $j \geq 3$, $\beta(S') = \beta(S)$ remains even.
  We consider three subcases based on the value of $R(x'_2)$.
  \begin{itemize}[leftmargin=8pt]
   \item When $R(x'_2) \geq 2$: \cnum{10}\\
     Since $1\leq t\leq k $, $R(x'_1 - x'_2) = R(x_2 - (x_1 - t)) = R(t) \neq 0$, $(x'_1, x'_2)$ is not $k$-nice.  
     By the induction hypothesis, $S' \notin \mathbb{ P}$.
   \item When $R(x'_2) = 1$:  \cnum{11}\\
     Since $R(x_2) \geq 2$, we have $R(x'_1) = R(x_2) \neq 0$.  
     Hence, $(x'_1, x'_2)$ is not $k$-nice, and by the induction hypothesis, $S' \notin \mathbb{P}$.
   \item When $R(x'_2) = 0$: \cnum{12}\\
     Since $R(x_2) \geq 2$, we have $R(x'_1) = R(x_2) \neq 1$.  
     Thus, $(x'_1, x'_2)$ is not $k$-nice, and by the induction hypothesis, $S' \notin \mathbb{P}$.
  \end{itemize}
\item When $x_1 - t = 0$:
\begin{itemize}[leftmargin=8pt]
   \item When $x_3 = 0$: \cnum{13}\\ 
     In this case, we have $x'_1 = x_2$ and $x'_j = 0$ for all $j \geq 2$.  
     Since $R(x'_2) = 0$ and $R(x'_1) = R(x_2) \geq 2$, $(x'_1, x'_2)$ is not $k$-nice.  
     By the induction hypothesis, $S' \notin \mathbb{P}$.
  \item When $x_3 = 1$: \cnum{14}\\
    In this case, $x_j'=x_{j+1}$ for $1 \leq j \leq n-1$, and $\beta(S') = \beta(S) - 1$ is odd.
    Moreover, since $R(x_1)=R(x_2)\geq 2$, we have $R(x'_1 - x'_2) = R(x_2 - x_3)\neq 0$.
    By the induction hypothesis, $S' \notin \mathbb{P}$.
\end{itemize}
\end{itemize}
\item When $R(x_1) =1$ and $R(x_2) = 0$:\\
In this case, $x_1 - x_2 = 1$ holds by the assumption $x_1 - x_2 < t \leq k$.
There are two possibilities.
\begin{itemize}
\item When $t=1$: \cnum{15}\\
This is a stable move.
\item When $2 \leq t \leq \min\{x_1,k\}$: \cnum{16}\\
By $x_1 \geq t \geq 2$ and $R(x_1)=1$, we have $x_1 \geq k + 2$.
Then, $x_1 - t \geq x_1 - k \geq 2 > x_3 $.
Thus, $x_1'=x_2$, $x_2'=x_1-t$, and $x_j'=x_j$ for all $j \geq 3$ , which implies that $\beta(S')=\beta(S)$ is even.
By $R(x_1)=1$ and $k \geq t \geq 2$, we have $R(x_2')=R(x_1-t)\geq 2$.
But $R(x'_1 - x'_2) = R(x_2 - (x_1 - t)) = R(-1+t)\neq 0$.
Thus, $(x'_1, x'_2)$ is not $k$-nice, by the induction hypothesis, $S' \notin \mathbb{P}$.
\end{itemize}
\end{enumerate}

\vspace{4pt} \noindent
\textbf{Case 2: $\beta(S)$ is odd}.

\noindent
We now consider the case where $\beta(S)$ is odd.
We will prove that if $R(x_1-x_2)=0$, then the position $S$ is a P-position. 
In the case where $x_1-x_2 \geq k+1$, removing any number of stones is a stable move, so $S'\notin \mathbb{P}$ by Lemma~\ref{claim3} and the induction hypothesis.

Now suppose $x_1 = x_2$.
Since $\beta(S)$ is odd and this subsection assumes $x_3 \leq 1$, we have $x_3=1$.
From this point on, the proof is based on the value of $x_1 - t$.
\begin{itemize}
  \item When $x_1-t \geq 1$: \cnum{17}\\
  We have $x_1'=x_2$, $x_2'=x_1-t$, and $x'_j = x_j$ for all $j \geq 3$.  
  Therefore, $\beta(S') = \beta(S)$ is odd.  
  Since $R(x'_1 - x'_2) = R(x_2 - (x_1 - t)) = R(x_1 - x_2 + t) = t \neq 0$,  
  by the induction hypothesis, $S' \notin \mathbb{P}$.
  \item When $x_1 - t = 0$: \cnum{18}\\
  Since all stones in the leftmost heap are removed, we have $x'_j = x_{j+1}$ for $1 \leq j \leq n-1$.  
  Thus, $\beta(S') = \beta(S) - 1$ is even, and $R(x'_2) = R(x_3) = R(1) = 1$.  
  However, since $x_2 = x_1 = t \leq k$, it follows that $R(x'_1) = R(x_2) \neq 0$.  
  Therefore, $(x'_1, x'_2)$ is not $k$-nice.
  Hence, by the induction hypothesis, $S' \notin \mathbb{P}$.
\end{itemize}

Therefore, when $x_3\leq1$, all positions satisfying the conditions of Theorem~\ref{theorem3} are P-positions.

%% file: docs/Bounded2.tex
\subsubsection{Case: when $x_3\geq2$}
Hereafter, we prove Theorem~\ref{theorem3} for $x_3\geq2$. 
The state transition diagram when $x_3 \geq2$ is shown in Figure \ref{fig:fig2}.
The proof is very similar to the one in normal play made by Xu and Zhu~\cite{xu2018bounded}, 
but it should be noted that the necessary and sufficient conditions for the follower position $S'$ to be a P-position changes depending on whether $x'_3\geq 2$ or not.\\
\begin{figure}[H]
    \includegraphics[width=1.0\columnwidth, clip]{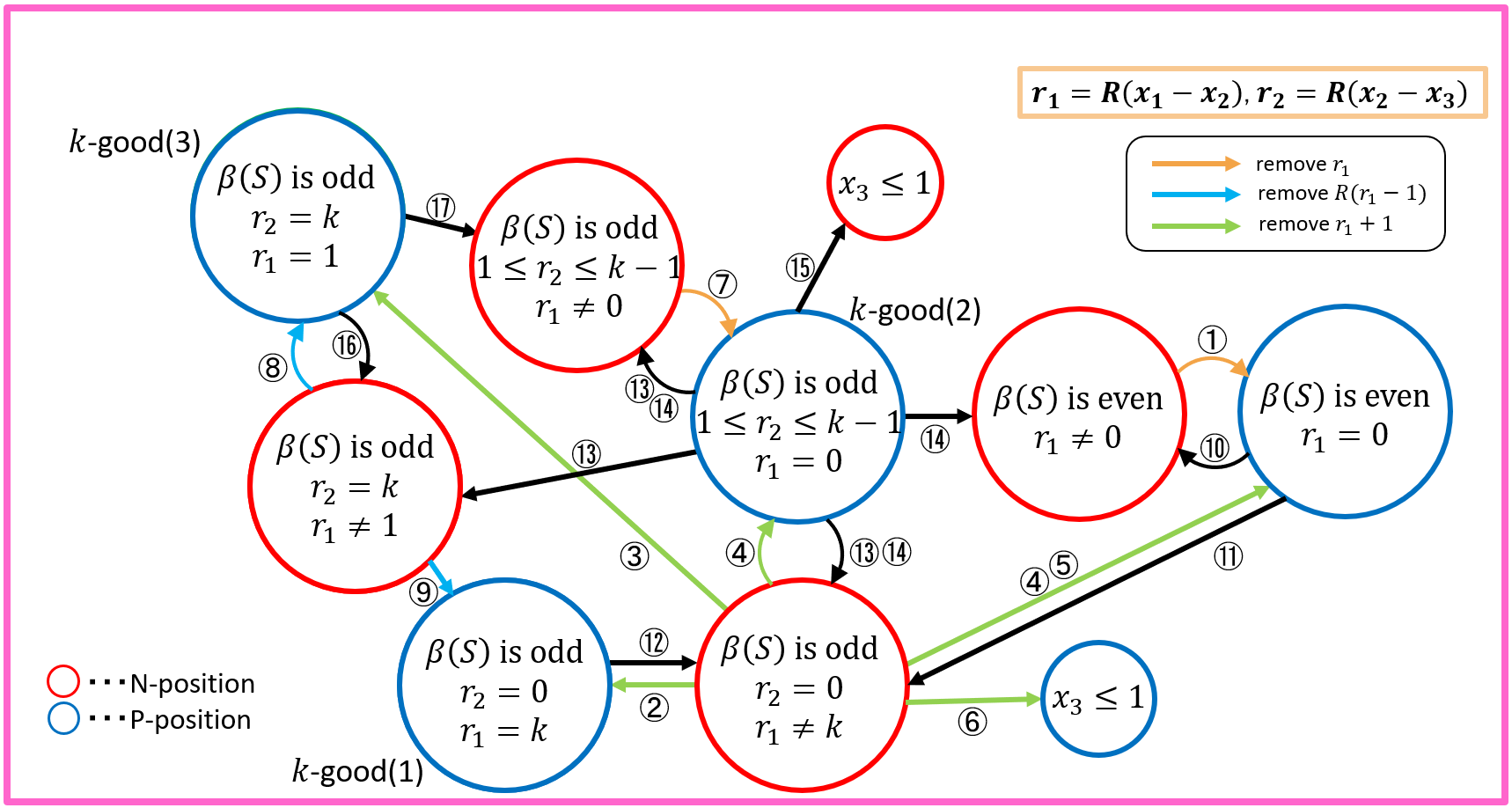}
    \caption{The state transition diagram when $x_3\geq2$}
    \label{fig:fig2}
\end{figure}

\vspace{4pt}\noindent
\textbf{Necessity.}
(The position $S$ is a P-position $\implies$ The position  $S$ satisfies one of the conditions of Theorem~\ref{theorem3})

\vspace{4pt}\noindent
\textbf{Case 1: $\beta(S)$ is even.} \cnum{1}

\noindent
We prove that if $S$ is a P-position, then $r_1=R(x_1 - x_2) = 0$, by contraposition.
From Fact~\ref{claim4}, if $r_1\neq 0$, then by removing $r_1$ stones from the position $S$, we obtain a position $S'$ such that $\beta(S') = \beta(S)$ is even, and $R(x_1' - x_2') = 0$ holds.
Therefore, by the induction hypothesis, we have $S' \in \mathbb{P}$ and thus $S$ is an N-position.

\vspace{4pt}\noindent
\textbf{Case 2: $\beta(S)$ is odd.}

\noindent
By contraposition, we will prove that $(x_1,x_2,x_3)$ is $k$-good if the position $S$ is a P-position. 
We will show that removing $t$ stones from a position $S$ that is not $k$-good is a winning strategy where $t$ is determined below.
\begin{enumerate}
    \item If $R(x_1-x_2)\neq k$ and $R(x_2-x_3)=0$, then $t=r_1+1$.
    
    Since $0 \leq r_1 \leq k$ and the assumption ensures $r_1 \neq k$, it follows that $1 \leq r_1 + 1 \leq k$.  
    Additionally, $x_1 \geq r_1 + x_2 \geq r_1 + x_3 \geq r_1 + 2$ holds.  
    Therefore, $1\leq t\leq \min \{x_1,k\}$ holds.
    \item If $R(x_1-x_2)\neq 0$ and $1\leq R(x_2-x_3)\leq k-1$, then $t=r_1$.
    
    Since $0 \leq r_1 \leq k$ and the assumption ensures $r_1 \neq 0$, it follows that $1 \leq r_1 \leq k$.
    Additionally, $x_1 \geq r_1 + x_2 > r_1 + x_3 \geq r_1 + 2$ holds.
    Therefore, $1\leq t\leq \min \{x_1,k\}$ holds.
    \item If $R(x_1-x_2)\neq 1$ and $R(x_2-x_3)=k$, then $t=R(r_1-1)$.
    
    Since the assumption ensures $r_1 \neq 1$, it follows that $1 \leq R(r_1 - 1) \leq k$.  
    Additionally, since $x_1 \geq x_2 \geq x_3+k \geq 2+k >R(r_1 - 1)$ holds,  
    Therefore, $1\leq t\leq \min \{x_1,k\}$ holds.
\end{enumerate}

Next, we prove that the position $S'$ after each of these moves is a P-position.
\begin{enumerate}
    \item When $R(x_1-x_2)\neq k$ and $R(x_2-x_3)=0$: 
    \begin{itemize}
        \item When $x_1-x_2\geq k$: \cnum{2}\\
        Since removing $t$ stones is a stable move, $x'_1 = x_1 - t=x_1-r_1-1$, $x'_j = x_j$ for all $j \geq 2$ and $\beta(S') = \beta(S)$ remains odd.  
        In addition, $x_3'=x_3\geq2$.
        Since $R(x'_1 - x'_2) = R(x_1 - x_2 - r_1 - 1) = R(x_1 - x_2 - R(x_1 - x_2) - 1) = k$ and $R(x'_2 - x'_3) = R(x_2 - x_3) = 0$, 
        $(x'_1,x'_2,x'_3)$ satisfies (1) of the definition of $k$-good.
        By the induction hypothesis, we conclude that $S' \in \mathbb{P}$.
        \item When $x_1-x_2<k$ and $x_2>x_3$: \cnum{3}\\
        Since $r_1=R(x_1-x_2)=x_1-x_2$, we have 
         $ x_1 - t = x_1 - (r_1 + 1) = x_1 - (x_1 - x_2 + 1) = x_2 - 1 \ge x_3$.
        Thus, $x_1' = x_2$, $x_2' = x_1-t = x_2-1$, and $x'_j = x_j$ for all $j \geq 3$.
        Here, $\beta(S') = \beta(S)$ remains odd.
        As $R(x'_1 - x'_2) = R(x_2 - x_2+1) = 1$ and $R(x'_2 - x'_3) = R(x_2 - x_3-1)=k$,
        $(x'_1,x'_2,x'_3)$ satisfies (3) of the definition of $k$-good.
        By the induction hypothesis, we conclude that $S' \in \mathbb{P}$.
        \item When $ x_1 - x_2 < k$, $ x_2 = x_3 $ and $(x_3>2$ or $x_4 >1)$. \\
        Since $r_1=x_1-x_2$, we have $x_1 - t = x_1 - (r_1 + 1) = x_2 - 1 = x_3 - 1$.
        So, $x_1'=x_2'=x_2=x_3$. 
        On the other hand, $x_3' = \max\{x_4,x_1-t\} = \max\{x_4,x_3-1\} \ge 2$ by the assumption. 
        \begin{itemize}[leftmargin=8pt]
        \item When $\beta(S)=1$. \cnum{4}\\
        In this case, $x_2=x_3 > x_4$ and thus $x_3' = \max\{x_4,x_3-1\} = x_3-1$.
        Since $R(x'_2-x'_3)=R(x_3-x_3+1)=1$ and $R(x'_1-x'_2)=0$,
        $(x'_1,x'_2,x'_3)$ satisfies (2) of the definition of $k$-good.
        Whether $\beta(S')$ is even or odd, the position $S'$ satisfies the condition of Theorem~\ref{theorem3} to be a P-position.
        By the induction hypothesis, we conclude that $S' \in \mathbb{P}$.
        \item When $\beta(S)\geq 3$. \cnum{5}\\
        In this case, $x_3=x_4$ implies $x_3'=x_4=x_3 > x_1 - t = x_3-1$.
        Therefore, $\beta(S')=\beta(S)-1$ is even.
        Since $R(x'_1-x'_2)=0$, by the induction hypothesis, we conclude that $S' \in \mathbb{P}$.
        \end{itemize}
        \item When $ x_1 - x_2 < k$, $ x_2 = x_3 = 2 $, and $ x_4 \leq 1 $: \cnum{6}\\
        Since $ r_1 = x_1 - x_2 $, 
        we have $ x_1 - t = x_2 - 1 = 1 $.  
        Therefore $ x'_1 = x'_2 = 2 $ and $ x'_3 \leq 1 $ hold.  
        Here,  $ R(x'_1 - x'_2) = R(2 - 2) = 0 $.  
        Also, since $ k \neq 1 $, we have $ R(x'_2) = R(2) = 2 $,  
        so $ (x'_1, x'_2) $ satisfies (2) of the definition of $k$-nice.  
        Regardless of whether $ \beta(S') $ is odd or even, the position $ S' $ satisfies the condition of Theorem~\ref{theorem3} for being a P-position.  
        Thus, by the induction hypothesis, we conclude that $ S' \in \mathbb{P} $.
    \end{itemize}
    \item When $R(x_1-x_2)\neq 0$ and $1\leq R(x_2-x_3)\leq k-1$: \cnum{7}\\
    Since $t = r_1$, by Fact~\ref{claim4}, $\beta(S') = \beta(S)$ is odd, and $x_3' = x_3 \geq 2$ holds.  
    Also, since $R(x_1' - x_2') = 0$ and $R(x_2' - x_3') = R(x_2 - x_3)$ hold, $(x_1', x_2', x_3')$ satisfies (2) of the definition of $k$-good.  
    Therefore, by the induction hypothesis, $S' \in \mathbb{P}$.
    \item When $R(x_1-x_2)\neq 1$ and $R(x_2-x_3)=k$: \\
    Since $ x_2 - x_3 \geq k $, 
    $x_1-t \geq x_3$.
    Thus, we have $ x'_j = x_j $ for all $j \geq 3$.  
    Therefore, $ \beta(S') = \beta(S) $ is odd, and $ x'_3 = x_3 \geq 2 $ holds.  
%
    \begin{itemize}
        \item When $ x_1 - x_2 \geq 2$: \cnum{8}\\
        If $r_1 = 0$, $x_1 - x_2 \geq k+1 > t$.
        If $r_1 > 0$, $ x_1 - x_2 \geq r_1 > R(r_1 - 1) = t $.
        In either case, removing $ t $ stones is a stable move and
        thus $ x'_1 = x_1 - R(r_1 - 1) $ and $x'_2 = x_2$ hold.  
        Here, $ R(x'_1 - x'_2) = R(x_1 - x_2 - R(r_1 - 1)) = 1 $  and  $ R(x'_2 - x'_3) = R(x_2 - x_3) = k $,  
        so $ (x'_1, x'_2, x'_3) $ satisfies (3) of the definition of $k$-good.
        By the induction hypothesis, we conclude that $ S' \in \mathbb{P} $.
        \item When $ x_1=x_2 $: \cnum{9}\\
        In this case, $ x'_1 = x_2 $ and $ x'_2 = x_1 - k $ hold.
        Here,  $ R(x'_2 - x'_3) = R(x_1 - k - x_3) = R(x_2 - x_3 - k) = 0 $ and $ R(x'_1 - x'_2) = R(x_2 - (x_1 - k)) = R(k) = k $,  
        so $ (x'_1, x'_2, x'_3) $ satisfies (1) of the definition of $k$-good.
        By the induction hypothesis, we conclude that $ S' \in \mathbb{P} $.
    \end{itemize}
\end{enumerate}

\vspace{4pt}\noindent
\textbf{Sufficiency.}
(The position $S$ satisfies one of the conditions of Theorem~\ref{theorem3} $\implies$ The position $S$ is a P-position)

If $S$ satisfies one of the conditions of Theorem~\ref{theorem3} and the move is stable, by Lemma~\ref{claim3} and the induction hypothesis, $S'$ is an N-position and thus $S$ is a P-position.
Therefore, it suffices to discuss the case where the move is not stable, i.e., $t > x_1-x_2$.

\vspace{4pt}\noindent
\textbf{Case 1: $\beta(S)$ is even.}

\noindent
We will show that if $ R(x_1 - x_2) = 0 $, the position $ S $ is a P-position.  
In position $ S $, when $ x_3 \geq 2 $ and $ \beta(S) $ is even,  
it follows that $ x_3 = x_4 \geq 2 $,  
so in the follower position $ S' $, we always have $ x'_3 \geq 2 $.  
The assumption $t > x_1-x_2$ implies $x_1 = x_2$.
There are two possible cases to consider.
\begin{enumerate}
    \item When $ x_1-t\geq x_3$: \cnum{10}\\
    In this case, $x_1' = x_2$, $x_2' = x_1-t$, and $ x'_j = x_j $ for all $ j \geq 3$,
    where $ \beta(S') = \beta(S) $ is even.  
    Here, $ R(x'_1 - x'_2) = R(x_2 - (x_1 - t)) = R(x_2 - x_1 + t) = t \neq 0 $,  
    so by the induction hypothesis, $ S' \notin \mathbb{P} $.
    \item When $ x_1-t < x_3$: \cnum{11}\\
    Since $ x_3 = x_4 $ holds, we have $ x'_j = x_{j+1} $ for $ 1 \leq j \leq 3 $.  
    Thus, $ \beta(S') = \beta(S) - 1 $ is odd.  
    Here,  $ R(x'_2 - x'_3) = R(x_3 - x_4) = 0 $.  
    Also, since $ x_1 - t < x_3 $,  it follows that $ x'_1 - x'_2 = x_2 - x_3 = x_1 - x_3 < t \leq k $,  
    so $ R(x'_1 - x'_2) \neq k $.  
    Therefore, $ (x'_1, x'_2, x'_3) $ is not $k$-good.
    By the induction hypothesis, we conclude that $ S' \notin \mathbb{P} $.
\end{enumerate}

\vspace{4pt}\noindent
\textbf{Case 2: $\beta(S)$ is odd.}

\noindent
We will prove that if $(x_1, x_2,x_3)$ is $k$-good, then the position $S$ is a P-position. 
We will show that $S'\notin \mathbb{P}$ if $(x_1, x_2)$ satisfies (1), (2) or (3) of the definition of $k$-good, in that order.
By the assumption $x_1 - x_2 < t \leq k$, we have $r_1 =x_1 - x_2$.
\begin{enumerate}
    \item When $R(x_2-x_3)=0$ and $x_1-x_2=k$: \cnum{12}\\
    In this case, any move is stable.
    \item When $1\leq R(x_2-x_3)\leq k-1$ and $x_1-x_2=0$: \\
    In this case, $ x'_1 = x_2=x_1$ holds.
    \begin{itemize}
        \item When $x_1-t\geq x_3$: \cnum{13}\\
        In this case, $ x'_2 = x_1 - t $ and $ x'_j = x_j $ for all $j \geq  3$.
        It follows that $ x'_3 = x_3 \geq 2 $ and $ \beta(S') = \beta(S) $ is odd.
        It remains to show that $(x'_1,x'_2,x'_3)$ is not $k$-good, which implies $S' \notin \mathbb{P}$ by the induction hypothesis.
        Since $R(x'_1-x'_2)=R(x_2-(x_1-t))=t\neq 0$, $(x'_1,x'_2,x'_3)$ does not satisfy (2) of the definition of $k$-good.
        To satisfy (1) or (3) of the definition of $k$-good, it is necessary that $R(x_1'-x_2')+R(x_2'-x_3') \in \{k,k+1\}$, i.e., $R(x_1'-x_3') \in \{0,k\}$.
        However, $R(x'_1-x'_3) = R(x_2-x_3) \notin \{0,k\}$ by the assumption.  
%
%
%
%
        \item When $ 2 \leq x_1-t < x_3$ or $x_1-t < 2 \leq x_4$: \cnum{14}\\
        In both cases, since $ x_1 - t < x_3 $, we have $ x'_2 = x_3 $.  
        Additionally, $ x'_3 \geq 2 $.  
        Here, $ R(x'_1 - x'_2) = R(x_2 - x_3) \neq 0 $.
        If $\beta(S')$ is even, then $S' \not\in \mathbb{P}$ by the induction hypothesis.
        Assume $\beta(S')$ is odd.
        It remains to show that $(x'_1,x'_2,x'_3)$ is not $k$-good.
        Since $x_1'=x_2$ and $x_2'=x_3$, we have $1 \leq R(x_1'-x_2') \leq k-1$.
        Thus, neither (1) nor (2) of the definition of $k$-good is satisfied.
        On the other hand, noting that $x_1 = x_2 > x_3$,
        $0 \leq x_2'-x_3' = x_3 - \max\{x_4, x_1 - t\} \leq x_3 - (x_1 - t) \leq t - 1 < k$ holds.
        That is, (3) of the definition of $k$-good is not satisfied either.
        \item When $x_1-t <2$ and $x_4<2$: \cnum{15}\\
        Also in this case, $ x'_2 = x_3 $ holds because $ x_1 - t < 2 \leq x_3 $. 
        Additionally, since $ x_4<2 $, it follows that $ x'_3 \leq 1 $.  
        Moreover, since $ x_1 - t<2 $, we have $ x_1 < t + 2 \leq k + 2 $ and $ 2 \leq x_3 < x_1 \leq k + 1 $.  
        Therefore, since $ k \geq 2 $, we obtain $ R(x'_2) = R(x_3) \geq 2 $.  
        Here, since $ R(x'_1 - x'_2) = R(x_2 - x_3) \neq 0 $, $ (x_1', x_2') $ does not satisfy (2) of the definition of $k$-nice.  
        Thus, regardless of whether $ \beta(S') $ is odd or even, by the induction hypothesis, $ S' \notin  \mathbb{P} $.
    \end{itemize}
    \item When $R(x_2-x_3)=k$ and $x_1-x_2=1$:
    \begin{itemize}
        \item When $ t=1$: \cnum{16}\\
        Since $t = x_1 - x_2$, this is a stable move.
%
        \item When $ 2\leq t \leq \min\{x_1, k\} $: \cnum{17}\\
        Since $ x_1 - x_3 \geq k + 1 $, we have $x_1' = x_2$, $x_2' = x_1 - t$, and $ x'_j = x_j $ for all $j \geq 3$.
        Here, $ x'_3 = x_3 \geq 2 $ and $ \beta(S') = \beta(S) $, which is odd.  
        It remains to show that $ (x'_1, x'_2, x'_3) $ is not $k$-good.
        Here, $R(x_2' - x_3') = R(x_1 - t - x_3) = R(x_2 + 1 - t - x_3) = R(-t)$ is within the range from $1$ to $k - 1$, and $R(x_1' - x_2') = R(x_2 - x_1 + t) = t - 1 \neq 0$.
        Therefore, $(x_1', x_2', x_3')$ is not $k$-good.    
        By the induction hypothesis, we conclude that $S' \notin \mathbb{P}$.
    \end{itemize}
\end{enumerate}

This completes the proof of Theorem~\ref{theorem3}.

%% file: docs/Greedy.tex
\section{Misère Greedy Nim}

\subsection{The Rule of Misère Greedy Nim}
The rules of the misère Greedy Nim are as follows.
The starting position consists of several heaps of stones, with each heap containing an arbitrary number of stones.
Two players take turns removing stones from the heap with the largest number of stones.

Obviously misère Greedy Nim is identical to misère $k$-bounded Greedy Nim if every heap has at most $k$ stones.
Therefore, using the result on misère $k$-bounded Greedy Nim presented in the previous section, one can establish a complete solution to the game.
This section gives a direct analysis on misère Greedy Nim.

\subsection{Theorem and Proof}
\begin{definition}
    For a position $S=(x_1,x_2,\dots,x_n)$, let
    \[
      \alpha(S)= 
          \begin{cases}
              0              &  \text{if $x_1=0$,}\\
              \max\{j:x_j=x_1\} &  \text{if $x_1 \neq 0$}, \nonumber
          \end{cases}
      \]
which is the number of  repetitions of the value of $x_1$ in the position $S$.
\end{definition}
The normal play Greedy Nim has been analyzed by Albert and Nowakowski~\cite{albert2004nim}.
\begin{theorem}[\cite{albert2004nim}]
\label{theorem4}
    A position $S$ is a P-position in normal play Greedy Nim if and only if $\alpha(S) $ is even.
\end{theorem}

We establish the misère play counterpart of Theorem~\ref{theorem4}.
\begin{theorem}
\label{theorem5}
    A position $S = (x_1, x_2, \dots, x_n)$ is a P-position in misère Greedy Nim if and only if one of the following holds:
    \begin{itemize}
        \item $x_1\leq 1$ and $\alpha(S)$ is odd.
        \item $x_1\geq 2$ and $\alpha(S)$ is even.
    \end{itemize}
\end{theorem}


\begin{proof}
One can rephrase the theorem as follows:
\def\labelenumi{(\theenumi)}
\begin{enumerate}
    \item If $x_1 \leq 1$, $S$ is a singular position; that is, the outcomes of the normal play and the misère play disagree.
    \item If $x_1 \geq 2$, $S$ is a standard position; that is, the outcomes of the normal play and the misère play are the same.
\end{enumerate}
We show the above by induction on the total number of stones in a position.

(1)
When $x_1 = 1$, the only possible operation for a position $S$ is to remove one heap. 
After removing the stones, $\alpha(S')$ and $\alpha(S)$ have different parities.
Obviously $S$ is a P-position when $\alpha(S)=1$ and then the claim follows inductively.

(2)
When $x_1 \geq 2$ and $x_2 \leq 1$, $\alpha(S)=1$ is odd.
Consider the two possible follower positions $S'$ and $S''$ obtained by removing all stones from the leftmost pile and removing all but one stone, respectively.
Both follower positions have no heaps with two or more stones and $\alpha(S'')=\alpha(S')+1$ holds.
From (1), one of these positions is a P-position.  
Therefore, $S$ is an N-position.

When $x_2 \geq 2$, every follower position of $S$ has at least one heap with at least 2 stones.  
By the induction hypothesis, for every follower position of $S$, the misère and normal outcomes are equal.  
Hence, the misère and normal outcomes of position $S$ are also equal.
\end{proof}

Theorem~\ref{theorem5} can also be derived by reformulating Theorem~\ref{theorem3} assuming that the constant $k$ is at least $x_1$.

%% file: docs/conclusion.tex
\section{Conclusion}

\subsection{Results}
In this paper, we present conditions that efficiently determine the outcome of the misère Bounded Greedy Nim and Greedy Nim, and provide proofs for these conditions.
For misère Bounded Greedy Nim, the proof is constructed by dividing game positions into those where the results of the normal play can be directly applied and those where they cannot.
For misère Greedy Nim, we conduct an analysis by completely identifying the singular and standard positions.
Furthermore, we show that the results for misère Greedy Nim can be derived using the results from misère Bounded Greedy Nim.

\subsection{Future work}
One of the future challenges is the development of algorithms that can efficiently compute the Grundy number of Greedy Nim and Bounded Greedy Nim.
The Grundy values, proposed by Sprague~\cite{sprague1935uber} and Grundy~\cite{grundy}, is a value that enables the determination of game outcomes. 
Moreover, Grundy values often reveal hidden properties of a game, making them extremely important for game analysis.
Although Grundy values can be computed recursively by tracing back from terminal positions, this approach requires an enormous amount of computation time. 
Therefore, finding ways to compute Grundy values efficiently without recursive calculation—and ideally, discovering closed-form expressions for them—is considered a major goal in the study of impartial games. 
Grundy values for Greedy Nim have been investigated by Schwartz ~\cite{Schwartz01111971}.
However, even for the two heap Greedy Nim, where the number of stones in the first heap is $a$ and that in the second heap is $b$, no algorithm has yet been discovered that computes the Grundy value in time $\log a + \log b$~\cite{albert2004nim}.
Moreover, we aim to analyze a variety of games for which a complete analysis has not yet been achieved in their misère play, such as Dawson's Chess~\cite{dawson1935caissa} and Treblecross (one-dimensional tic-tac-toe)~\cite{berlekamp2004winning}.

%% file: main.bbl
\begin{thebibliography}{10}

\bibitem{albert2004nim}
Michael~H. Albert and Richard~J. Nowakowski.
\newblock Nim restrictions.
\newblock {\em Integers: Electronic Journal of Combinatorial Number Theory}, 4(G01):G01, 2004.

\bibitem{berlekamp2004winning}
Elwyn~R. Berlekamp, John~H. Conway, and Richard~K. Guy.
\newblock {\em Winning ways for your mathematical plays, volume 4}.
\newblock AK Peters/CRC Press, 2004.

\bibitem{bouton1901nim}
Charles~L. Bouton.
\newblock Nim, a game with a complete mathematical theory.
\newblock {\em Annals of mathematics}, 3(1/4):35--39, 1901.

\bibitem{dawson1935caissa}
Thomas~Rayner Dawson.
\newblock {\em Caissa's Wild Roses}.
\newblock Number~1. TR Dawson, 1935.

\bibitem{dudeney2002canterbury}
Henry~Ernest Dudeney.
\newblock {\em The Canterbury Puzzles}.
\newblock Courier Corporation, 2002.

\bibitem{grundy}
P.~M. Grundy.
\newblock Mathematics and games.
\newblock {\em Eureka}, 2:6--8, 1939.

\bibitem{Grundy_Smith_1956}
P.~M. Grundy and C.~A.~B. Smith.
\newblock Disjunctive games with the last player losing.
\newblock {\em Mathematical Proceedings of the Cambridge Philosophical Society}, 52(3):527--533, 1956.

\bibitem{locke2021amalgamation}
S.~C. Locke and B.~Handley.
\newblock Amalgamation {Nim}.
\newblock {\em Integers: Electronic Journal of Combinatorial Number Theory}, 21, 2021.

\bibitem{PLAMBECK2008593}
Thane~E. Plambeck and Aaron~N. Siegel.
\newblock Misère quotients for impartial games.
\newblock {\em Journal of Combinatorial Theory, Series A}, 115(4):593--622, 2008.

\bibitem{Schwartz01111971}
Benjamin~L. Schwartz.
\newblock Some extensions of {Nim}.
\newblock {\em Mathematics Magazine}, 44(5):252--257, 1971.

\bibitem{sprague1935uber}
R.~Sprague.
\newblock \"{U}ber mathematische kampfspiele.
\newblock {\em Tohoku Mathematical Journal, First Series}, 41:438--444, 1935.

\bibitem{welter1952advancing}
C.~P. Welter.
\newblock The advancing operation in a special abelian group.
\newblock In {\em Indagationes Mathematicae (Proceedings)}, volume~55, pages 304--314. Elsevier, 1952.

\bibitem{wythoffgame}
W.~A. Wythoff.
\newblock A modification of the game of {Nim}.
\newblock {\em Nieuw Archief voor Wiskunde}, 7(2):199--202, 1907.

\bibitem{xu2018bounded}
Rongxing Xu and Xuding Zhu.
\newblock Bounded greedy {Nim}.
\newblock {\em Theoretical Computer Science}, 746:1--5, 2018.

\bibitem{yamasaki1980misere}
Yohei Yamasaki.
\newblock On misere {Nim}-type games.
\newblock {\em Journal of the Mathematical Society of Japan}, 32(3):461--475, 1980.

\end{thebibliography}
